\documentclass[conference,letterpaper]{IEEEtran}

\usepackage[utf8]{inputenc}
\usepackage[T1]{fontenc}
\usepackage{url}
\usepackage{ifthen}
\usepackage{cite}
\usepackage[cmex10]{amsmath}
\usepackage{graphicx}

\usepackage{amsfonts}
\usepackage{amssymb}
\usepackage{mathtools}
\usepackage{amsthm}

\newtheorem{theorem}{Theorem}
\newtheorem{lemma}{Lemma}
\newtheorem{definition}{Definition}

\usepackage{pgfplots}
\usepackage{tikz}
\pgfplotsset{compat=1.18}
\usetikzlibrary{arrows,decorations.pathreplacing,shapes.geometric}
\usepackage{circuitikz}

\usepackage[algo2e,ruled,vlined,linesnumbered]{algorithm2e}
\usepackage{algorithm}
\usepackage{algpseudocode}
\usepackage{comment}
\usepackage{xcolor}

\newcommand{\T}{\top}
\newcommand{\wt}{\mathrm{wt}}

\newcommand{\cA}{{\cal A}}

\newcommand{\cC}{{\cal C}}

\newcommand{\cE}{{\cal E}}

\newcommand{\cI}{{\cal I}}

\newcommand{\cK}{{\cal K}}

\newcommand{\cP}{{\cal P}}

\newcommand{\cR}{{\cal R}}
\newcommand{\cS}{{\cal S}}

\newcommand{\mbC}{{\mathbb C}}

\newif\ifarxiv
\arxivtrue

\IEEEoverridecommandlockouts

\begin{document}

\title{Stabilizer-Assisted Inactivation Decoding of Quantum Error-Correcting Codes with Erasures}

\author{%
\IEEEauthorblockN{%
    Giulio Pech%
        \IEEEauthorrefmark{1},
    Mert G\"{o}kduman%
        \IEEEauthorrefmark{1}%
        \IEEEauthorrefmark{2},
    Hanwen~Yao%
        \IEEEauthorrefmark{1}%
        \IEEEauthorrefmark{2},
    and Henry~D.~Pfister%
        \IEEEauthorrefmark{1}%
        \IEEEauthorrefmark{2}%
        \IEEEauthorrefmark{3}%
}
\IEEEauthorblockA{\IEEEauthorrefmark{1}%
                  Duke Quantum Center, Duke University, 
                  Durham, NC, USA
                  }
\IEEEauthorblockA{\IEEEauthorrefmark{2}%
                  Department of Electrical and Computer Engineering, 
                  Duke University, 
                  Durham, NC, USA}
\IEEEauthorblockA{\IEEEauthorrefmark{3}%
                  Department of Mathematics, Duke University, 
                  Durham, NC, USA}
}

\maketitle



\begin{abstract}
In this work, we develop a reduced complexity maximum likelihood (ML) decoder for quantum low-density parity-check (QLDPC) codes over erasures. Our decoder combines classical inactivation decoding, which integrates peeling with symbolic guessing, with a new dual peeling procedure. In the dual peeling stage, we perform row operations on the stabilizer matrix to efficiently reveal stabilizer generators and their linear combinations whose support lies entirely on the erased set. Each such stabilizer identified allows us to freely fix a bit in its support without affecting the logical state of the decoded result. This removes one degree of freedom that would otherwise require a symbolic guess, reducing the number of inactivated variables and decreasing the size of the final linear system that must be solved. We further show that dual peeling combined with standard peeling alone, without inactivation, is sufficient to achieve ML for erasure decoding of surface codes. Simulations across several QLDPC code families confirm that our decoder matches ML logical failure performance while significantly reducing the complexity of inactivation decoding, including more than a 20\% reduction in symbolic guesses for the B1 lifted product code at high erasure rates.

\end{abstract}


%

\vspace{-2mm}

\section{Introduction}

Quantum error correction is essential for quantum computing because physical qubits are very sensitive to noise
.
Surface codes~\cite{KITAEV20032,dennis2002topological} are often seen as the most practical near-term approach, but quantum low-density parity-check (QLDPC) codes
offer a promising approach to fault-tolerant quantum computation with 
higher rate and lower overhead.
A key challenge is that QLDPC codes require irregular and long-range connectivity between qubits.
Recently, a number of techniques have been proposed to overcome this challenge across several qubit platforms~\cite{henriet2020quantum,xu2024constant,monroe2014large,moses2023race,chen2024benchmarking,hong2024entangling,bravyi2024high}.

A wide variety of QLDPC constructions are now available.
The hypergraph product (HGP) construction~\cite{tillich2009quantum} is an early and foundational approach that combines two classical LDPC codes to obtain QLDPC codes with controlled rate and distance.
Other families include lifted product and balanced product codes 
\cite{Panteleev_2022_lifted,panteleev2022asymptotically,breuckmann2021balanced_product}, bivariate bicycle codes 
\cite{bravyi2024high}, and related generalizations 
\cite{Panteleev2021degeneratequantum,LinPryadko2024quantum_2BGA,WangLinPryadko2023quantum_2block}.

In this work, we focus on decoding QLDPC codes over the quantum erasure channel. In this noise model, each qubit either experiences no error or is replaced by an erasure flag and subjected to a uniform random Pauli error. This setting is useful in research for several reasons.
First, it simplifies analysis and provides intuition for more general noise models.
Second, both surface codes and QLDPC codes are known to achieve substantially higher thresholds under erasure noise than under Pauli noise, and recent erasure conversion proposals across several architectures can translate the dominant physical errors into erasures~\cite{wu2022erasure,sahay2023high,ma2023high,kang2023quantum,kubica2023erasure,teoh2023dual}, making this model increasingly relevant in practice.

On the classical side, maximum likelihood (ML) decoding of LDPC codes over erasures can be performed efficiently via peeling plus inactivation~\cite{shokrollahi2005inactivation,shokrollahi2005raptor,Lazaro2017Inactivation,Pishro2004peeling_guessing,measson2005maxwellconstructionhiddenbridge}. 
Peeling resolves erased variables using degree-1 checks on the Tanner graph, and when combined with inactivation, symbolic guesses are introduced whenever peeling stalls so that the process can continue. After peeling removes all resolvable variables, the symbolic guesses are determined by solving a small system of linear equations via Gaussian elimination (GE), yielding ML decoding.

On the quantum side, erasure decoding has been investigated for both topological codes and QLDPC codes. For surface codes, Delfosse and Z\'{e}mor~\cite{PhysRevResearch.2.033042} proved that linear time ML decoding is possible on the erasure channel. The union-find decoder, which can correct both errors and erasures, was initially proposed for topological codes~\cite{delfosse2021almost} and was later extended to QLDPC 
codes~\cite{delfosse2022toward}. 
For HGP codes, Connolly et al.~\cite{connolly2024fast} introduced a fast decoder that combines peeling with a vertical and horizontal decomposition of the residual graph. Iterative belief propagation (BP) with guided decimation~\cite{yao2023belief,gokduman2024erasure,R2025improved_BPGD_erasure} also performs well on HGP codes, exhibiting waterfall and error floor regimes~\cite{Gokduman2025hgp_waterfall}, and it can be applied to decode general QLDPC codes. A cluster decoder based on peeling and cluster decomposition has also been proposed. It can approach ML performance efficiently in the low erasure regime but may have high complexity near threshold~\cite{yao2024clusterdecompositionimprovederasure}.

In this paper, we apply peeling together with inactivation to achieve ML decoding of QLDPC codes over erasures. A key challenge in this setting is that many stabilizers, either individual generators or their linear combinations, may be supported entirely on the erased set. Although such fully erased stabilizers do not affect the logical state, they create stopping sets for peeling and can trigger unnecessary symbolic guesses during inactivation decoding, increasing the computational burden of both the peeling and inactivation steps and of the final GE stage.

To address this issue, we introduce a dual peeling procedure that efficiently identifies fully erased stabilizers through row operations on the stabilizer matrix. Our stabilizer-assisted inactivation decoder proceeds by first applying dual peeling to locate as many fully erased stabilizers as possible, then fixing one erased bit in the support of each of them to remove degrees of freedom that would otherwise appear as symbolic guesses. 

Peeling plus inactivation decoding is then applied to the reduced erasure pattern, and the symbolic guesses are determined by solving a system of linear equations through GE at the end. This integrated approach greatly reduces the number of symbolic guesses and the size of the final linear system, providing an efficient method for achieving ML erasure decoding for QLDPC codes. For surface codes in particular, we show that applying dual peeling followed by standard peeling, without invoking inactivation, is sufficient to reach ML decoding over erasures.

Simulations across several QLDPC code families confirm that our decoder attains the same logical failure performance as ML decoding while reducing the complexity.

\emph{Note:}
    This work will be presented as a poster at QIP 2026~\cite{pech2026stabilizer_qip}.
    Right before submission of the conference version~\cite{pech2026stabilizer_isit}, a similar independent work was posted to arXiv~\cite{freire2026quantum}. 

\section{Preliminaries}
\label{sec:prelim}

In this section, we briefly review standard material.
\ifarxiv
Appendix~\ref{app:notation}.
\fi
We also review standard classical and quantum coding concepts which can also be found in standard references such as~\cite{gottesman1997stabilizer,calderbank1997quantum,Nielsen_Chuang_2010,Richardson_Urbanke_2008,tillich2009quantum}.

\subsection{Classical Erasure Correction}

A binary linear code $\cC \subseteq \mathbb{F}_2^n$ is a subspace of the vector space $\mathbb{F}_2^n$.
A parity-check matrix $H \in \mathbb{F}_2^{m \times n}$ for $\cC$ satisfies $H x^\T = 0$ for all $x \in \cC$. The Tanner graph of $H$ is a bipartite graph with a variable node for each column of $H$, a check node for each row of $H$, an edge between variable $j$ and check $i$ whenever $H_{ij} = 1$.

Over the binary erasure channel (BEC), each transmitted bit is either received correctly or erased with probability $p$. The decoder knows which positions were erased. Erasure decoding can be formulated as solving the linear system
\begin{equation}\label{eq:syndrome-system}
  H_\cE x_\cE^\T = s,
\end{equation}
where $\cE$ is the set of erased positions, $H_{\cE}$ is the submatrix of $H$ formed by the erased columns, $x_{\cE}\in\mathbb{F}_2^{|\cE|}$ are the unknown erased bits, and $s\in\mathbb{F}_2^{m}$ is the measured syndrome.
The low-complexity peeling decoder iteratively uses degree-1 checks in the Tanner graph to solve for erased bits and remove them from the graph, until no degree-1 checks remain. If peeling terminates with unresolved erasures, the residual set of erasures is called a \emph{stopping set}.

\subsection{Stabilizer Formalism}

A single qubit in a pure quantum state is modeled as a unit vector in the Hilbert space $\mbC^2$, and an $n$-qubit system lives in $(\mbC^2)^{\otimes n}$. The Pauli matrices are given by
\[
I =
\begin{bmatrix}
1 & 0 \\ 0 & 1
\end{bmatrix},\,\,
X =
\begin{bmatrix}
0 & 1 \\ 1 & 0
\end{bmatrix},\,\,
Y =
\begin{bmatrix}
0 & -i \\ i & 0
\end{bmatrix},\,\,
Z =
\begin{bmatrix}
1 & 0 \\ 0 & -1
\end{bmatrix}.
\]
The $n$-qubit Pauli group $\cP_n$ is generated by $\kappa D(a,b)$ with
\[
D(a,b) := X^{a_1} Z^{b_1} \otimes X^{a_2} Z^{b_2} \otimes \cdots \otimes X^{a_n} Z^{b_n},
\]
with $a,b \in \mathbb{F}_2^n$, and $\kappa \in \{\pm 1,\pm i\}$.
The weight of a Pauli operator is the number of qubits on which it acts nontrivially
.

A commutative subgroup $\cS \subset \cP_n$ that does not contain $-I$ is called a \emph{stabilizer subgroup}.
An $[[n,k]]$ \emph{stabilizer code} with $n$ physical qubits and $k$ logical qubits is a $2^k$-dimensional subspace defined by
\[
\cC = \{ \, |\psi\rangle \in (\mbC^2)^{\otimes n} : M|\psi\rangle = |\psi\rangle,\; \forall M \in \cS \}.
\]
The operators in $\cS$ are called \emph{stabilizers}.

A \emph{logical Pauli} is an operator $L \in \cP_n \setminus \cS$ satisfying $L \cC = \cC$.
The \emph{distance} $d$ of a quantum code is the minimum weight among all logical Pauli operators.
A stabilizer code with distance $d$ can correct up to $d-1$ erasures, detect up to $d-1$ errors, and can correct up to $t=\lfloor (d-1)/2 \rfloor$ errors.

\subsection{CSS Codes and Quantum LDPC Codes}

CSS codes~\cite{Calderbank_1996} form an important class of stabilizer codes. They are defined by a pair of classical binary codes $\cC_1, \cC_2 \subseteq \mathbb{F}_2^n$ such that $\cC_2^\perp \subseteq \cC_1$. In matrix form, one specifies two binary matrices $H_X,H_Z$ satisfying
\begin{equation}
  H_X H_Z^{\T} = H_Z H_X^{\T} = 0 \pmod{2}.
\end{equation}
If $H_X$ and $H_Z$ are low-density matrices, then the resulting stabilizer code is a QLDPC code.

The hypergraph product (HGP) construction~\cite{tillich2009quantum} defines a CSS code using two classical parity-check matrices. If $H_1$ and $H_2$ are $m_1 \times n_1$ and $m_2 \times n_2$ binary matrices, then the HGP code has
\begin{align}
H_X &=
\begin{bmatrix}
H_1 \otimes I_{n_2} & I_{m_1} \otimes H_2^{\T}
\end{bmatrix}, \\
H_Z &=
\begin{bmatrix}
I_{n_1} \otimes H_2 & H_1^{\T} \otimes I_{m_2}
\end{bmatrix}.
\end{align}

\subsection{Syndrome Decoding on the Quantum Erasure Channel}

On the quantum erasure channel \cite{bennett1997capacities, grassl1997codes}, erasures are applied independently to the qubits with probability $p$.
Erased qubits are corrupted by a uniform random Pauli error and the set $\cE$ of erased qubits is known at the decoder.
We represent $X$-type and $Z$-type error patterns by binary vectors $x,z \in \mathbb{F}_2^n$.
The elements of these vectors are always zero for indices in the set $\cE^c$ of non-erasures and are Bernoulli$\big(\frac{1}{2}\big)$ for indices in the set $\cE$.
The measured syndromes $s_x , s_z \in \mathbb{F}_2^m$ are defined by $s_x = H_Z x$ and $s_z = H_X z$.
To decode, one sets up two systems of linear equations
\begin{equation}
  H_{Z,\cE} x_\cE = s_x, \quad H_{X,\cE} z_\cE = s_z.
\end{equation}

For a given erasure pattern and syndrome, there may be many Pauli errors consistent with the observations:
\begin{itemize}
  \item Errors that differ by an element of the stabilizer group $\cS$ are equivalent, have the same correction operator, and do not change the logical state;
  \item Errors that differ by a logical operator have different correction operators.
\end{itemize}
An ML decoder chooses the most likely error pattern consistent with the erasures and the syndrome. In the case of ties, it may either declare failure or break the tie arbitrarily. In practice, decoding may terminate in one of four ways:
\begin{itemize}
  \item \textbf{Success:} the estimated error equals the actual error,
  \item \textbf{Degenerate success:} the difference between the estimated and actual error is a stabilizer,
  \item \textbf{Logical failure:} the estimated error differs from the actual error by a logical operator,
  \item \textbf{Non-convergent:} the decoder does not recover a solution consistent with the measured syndrome.
\end{itemize}

\subsection{Inactivation Decoding}
\label{sec:peeling-inactivation}

\begin{figure}[t]
\centering
\scalebox{0.85}{%
\begin{tikzpicture}[scale=1.0,>=stealth]

  \def\W{6}        
  \def\Ws{0.15}    
  \def\H{4}

  \draw (0,0) rectangle (\W+\Ws,\H);

  \draw (\W,0) -- (\W,\H);
  \node[above] at (\W+0.5*\Ws,\H) {$s$};

  \draw (4.5,0) -- (4.5,\H);

  \draw (0,1.5) -- (\W+\Ws,1.5);

  \draw (0,\H) -- (4.5,1.5);

  \fill[gray!35] (4.5,0) rectangle (\W,\H);

  \node at (2.0,3.0) {$A$};
  \node at (2.0,0.6) {$B$};
  \node[white] at (5.25,2.7) {$D$};
  \node[white] at (5.25,0.7) {$C$};

  \fill[orange!35] (4.1,0) rectangle (4.25,\H);
  \draw[orange!70!black] (4.1,0) rectangle (4.25,\H);

  \draw[->,thick,orange!70!black]
    (4.175,3.75) .. controls (4.6,3.9) and (4.8,3.9) .. (5.15,3.6);
    \node[orange!70!black,align=left] at (5.45,3.65)
      {\scriptsize inactivate\\[-1pt]\scriptsize (guess)};

  \draw[<-|] (-0.7,\H) -- (-0.7,1.5);
  \draw[|->] (-0.7,1.5) -- (-0.7,0);
  \node[left] at (-0.7,2.75) {$\ell_r$};
  \node[left] at (-0.7,0.75) {$m-\ell_r$};

  \draw[<-|] (\W+\Ws+0.7,\H) -- (\W+\Ws+0.7,0);
  \node[right] at (\W+\Ws+0.7,2) {$m$};

  \draw[<-|] (0,4.4) -- (4.5,4.4);
  \draw[|->] (4.5,4.4) -- (\W,4.4);
  \node[above] at (2.25,4.4) {$\ell_r$};
  \node[above] at (5.25,4.4) {$\ell_x$};

  \draw[<-|] (0,-0.4) -- (\W,-0.4);
  \node[below] at (3,-0.4) {$k$};
\end{tikzpicture}}
\vspace{-2mm}
\caption{Block structure of the augmented erased submatrix $H_{\cE}$ during inactivation decoding.
The $k$ erased columns are split into $\ell_r$ active columns processed by peeling (left) and $\ell_x$ inactivated/guessed columns (gray band, right). The last column $s$ is the syndrome, used to create the augmented matrix to solve \eqref{eq:syndrome-system}.
Blocks $C$ and $D$ contain the coefficients of the $\ell_x$ guessed variables in these two sets of constraints, $C$ to solve for the inactive ones, and $D$ to determine the active bits given the inactive ones.}
\label{fig:divided-matrix}
\vspace{-3mm}
\end{figure}
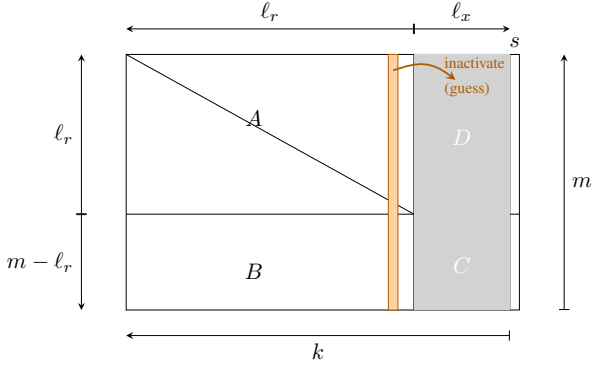

Inactivation decoding augments peeling by allowing the decoder to \emph{temporarily treat} certain erased variables as formal unknowns so that peeling can continue \cite{shokrollahi2005inactivation, shokrollahi2005raptor, Lazaro2017Inactivation, measson2005maxwellconstructionhiddenbridge} .
Formally, one can maintain a partition of the erased indices
\[
\cE = \cA \,\cup \, \cI,
\]
where $\cA$ are \emph{active} erased variables and $\cI$ are \emph{inactivated} (guessed) variables.
Let $g\in\mathbb{F}_2^{|\cI|}$ denote the vector of inactivated variables, with an arbitrary but fixed ordering of $\cI$.
We interpret ``making a symbolic guess'' for $x_j$ as removing $j$ from $\cA$, adding $j$ to $\cI$, and introducing a new formal variable $g_{|\cI|}$ with the identification $g_{|\cI|} \equiv x_j$.
Thus, subsequent peeling operations treat $x_j$ as  known 
(it is removed from $\cA$) and finding its true value is deferred to a later stage.

Operationally, the decoder alternates the following steps until peeling has solved all the active variables:
\begin{enumerate}
\item \textbf{Peel:} While there exists a check row whose restriction to active variables has weight one, solve that active variable and remove it from $\cA$ (updating neighboring checks by substitution). Each solved variable is recorded as an affine function of $(s,g)$.
\item \textbf{Inactivate:} If peeling stalls (no active row of weight one), choose an index
\begin{equation}
  j^* \in \arg\max_{j\in\cA} c_j,
\end{equation}
where $c_j$ is the residual weight of column $j$. Then, create a new symbolic guess $g_{\ell}$, and move $j^*$ from $\cA$ to $\cI$.
\end{enumerate}

At termination, every erased variable is expressed as a function of $s$ and $g$.
Substituting these expressions into the remaining unsatisfied check equations yields a linear system involving only the guessed variables,
\begin{equation}
  C g = b,
\end{equation}
which is solved by Gaussian elimination over $\mathbb{F}_2$.
Back-substitution then recovers all erased bits.
Inactivation confines dense elimination to the $|\cI|\times|\cI|$ core, and on the BEC it achieves maximum-likelihood decoding while keeping $|\cI| \leq |\cE|$. Note that peeling can be viewed as the special case of inactivation decoding with no inactivations, in which case Gaussian elimination is never needed.


\section{Quantum Decoding}
\label{sec:method}

We now describe the quantum decoder studied in this work.
Because CSS decoding on the quantum erasure channel reduces to two independent binary erasure decoding problems, we describe the procedure for decoding $X$ errors; decoding $Z$ errors is identical with
$H_Z$ replaced by $H_X$.

Given an erasure set $\cE$ and syndrome $s_x$, the goal is to find $x_\cE\in\mathbb{F}_2^{|\cE|}$ satisfying
$H_{Z,\cE}x_\cE=s_x$, up to stabilizer equivalence.
Our decoder has two stages:
\begin{enumerate}
  \item \textbf{Dual-peeling preprocessing:} run the dual peeling procedure to uncover fully erased stabilizers and fix one erased degree of freedom per independent fully erased stabilizer~\cite{connolly2024fast, Stace-physreva10}; 
  \item \textbf{Classical inactivation decoding:} run the inactivation decoder on the reduced erased system.
\end{enumerate}
Intuitively, dual peeling reduces the number of degrees of freedom that must be represented by symbolic guesses, so the
Gaussian-elimination (GE) core in inactivation is smaller.

\subsection{Stage 1: Dual peeling and primal peeling}

If a stabilizer acts only on erased qubits, it can flip all variables in its support without changing the measured syndrome or the logical state.
Thus, in its support, we can fix one variable \emph{arbitrarily} and reduce the stabilizer degeneracy by one.
This idea was introduced and called pruning in~\cite{connolly2024fast} and is closely related to gauge fixing ideas~\cite{Stace-physreva10}.




Let $\cE$ be the erased qubits and $\cK=\cE^c$ the known qubits.
When decoding $X$ errors on the erasure channel, the unknown error vector is supported on $\cE$ and must satisfy
\[
H_{Z,\cE}x_\cE=s_x.
\]
Two solutions are physically equivalent if they differ by an $X$-type stabilizer: for any $u$,
\[
x_\cE \sim x_\cE \oplus (u^\T H_X)_\cE,
\]
and the $Z$-syndrome is unchanged because $H_Z H_X^{\T}=0$.
Hence, if there exists a nontrivial combination $u^\T H_X$ supported entirely on $\cE$, then we may \emph{fix} one erased variable arbitrarily and absorb any discrepancy by multiplying the recovery by this fully erased $X$-stabilizer.
Our goal in the dual peeling stage is therefore to identify such fully erased $X$-stabilizers by performing row operations on $H_X$ guided by the known set $\cK$.

We therefore run \emph{dual peeling} on the $X$-stabilizer matrix $H_X$, using the known set $\cK$
only to guide row operations. Dual peeling iteratively applies two simplification rules on the \emph{known} submatrix
to uncover \emph{fully erased stabilizers} and use them to remove unknowns for free. To ease the description, we reorder the $H_X$ matrix as
\[[\,H_{X,i,\cK}\;\; H_{X,i,\cE}\,],\]
and define its \emph{known-degree} as $d_\cK(i)=\mathrm{wt}(H_{X,i,\cK})$.

\begin{enumerate}
\item \textbf{Known-column weight-2 elimination.}
\label{rule1}
If there exists a known qubit $j\in\cK$ whose column has weight two in $H_{X,\cK}$,
incident to rows $i_1,i_2$, then remove $i_2$ and replace $i_1$ by their XOR:
\[
H_{X,i_1}\leftarrow H_{X,i_1}\oplus H_{X,i_2}.
\]

Since column $j$ has ones only in rows $i_1,i_2$, for any coefficient vector $u$ we have
$(u^\T H_X)_j=u_{i_1}\oplus u_{i_2}$.
If $u^\T H_X$ is fully erased then $j\in\cK$ implies $(u^\T H_X)_j=0$, hence $u_{i_1}=u_{i_2}$.
Therefore, every fully erased combination uses $H_{X,i_1}$ and $H_{X,i_2}$ only through their sum
$H_{X,i_1}\oplus H_{X,i_2}$.

\item \textbf{Row peeling on known-degree-1 rows.}
If there exists a row $i$ with $d_\cK(i)=1$, let $j$ be its unique neighbor in $\cK$.
Use row $i$ as a pivot to eliminate column $j$ from every other row that contains it:
for each $k\neq i$ with $H_{X,kj}=1$, set
\[
H_{X,k}\leftarrow H_{X,k}\oplus H_{X,i}.
\]
This “peels” known qubits out of the stabilizer description and can reduce the known-degree of other rows.
\end{enumerate}

The procedure repeats until neither rule applies. Any row $i$ with $d_\cK(i)=0$ at termination
is a \emph{fully erased $X$-stabilizer}: its support lies entirely in $\cE$.
Such a stabilizer can be used to fix one erased variable arbitrarily for free. For example, say that we pick one of the qubits inside the stabilizer to be $x_j =0$.
If the true error has $x_j=1$, we can multiply the recovery by this fully erased stabilizer,
which flips $x_j$ (and only other erased positions) without changing the syndrome or the logical state.
Thus each independent fully erased stabilizer eliminates one degree of freedom that would otherwise appear as a symbolic guess in the inactivation core, reducing the size of the matrix on which we apply Gaussian elimination.

In practice, we arbitrarily choose a qubit covered by the stabilizer and fix its value to 0.
The qubit is then added to the set of non-erased qubits (i.e., moved to the first part of the stabilizer matrix to be included in the peeling).
This ensures that each fixed bit removes exactly one degree of freedom.

\paragraph*{Remark (alternating primal and dual peeling)}

One might consider running rounds of dual peeling after each inactivation (or after each burst of primal peeling).
Empirically, for all code instances and erasure rates we tested, this did not further reduce the dimension of the residual GE core beyond what was achieved by a single initial round of dual-peeling pass.
We also have some theoretical understanding of why this happens.
Therefore, we only consider a single round of dual-peeling followed by inactivation decoding.

\begin{lemma}[No gain from repeating dual peeling after primal]
\label{lem:no_second_dual_pass}
Fix $\cE$ and $\cK=\cE^c$.  After running dual peeling on $(H_X,\cK)$ to exhaustion, re-running rule-2 dual peeling after any
amount of primal peeling cannot produce a new fully erased $X$-stabilizer.
\end{lemma}

\begin{proof}
Assume dual peeling is first run to exhaustion so that all rows have known-degree of at least 2.
Then, primal peeling acts on the syndrome system $H_{Z,\cE}x_\cE=s_x$ by solving some coordinates in $\cE$ and moving them to $\cK$.
When returning to dual peeling, newly marked columns as known can only increase the known-degree of a row.
Thus, the set of admissible dual-peeling moves is
unchanged by primal peeling and a second dual peeling cannot create any new row with $d_\cK(i)=0$.
\end{proof}

\subsection{Stage 2: Inactivation decoding}
After the preprocessing stage, we run the standard inactivation decoder on the (reduced) erased system.
When primal peeling stalls, we inactivate (symbolically guess) an erased variable (in our implementation, we choose a column of maximum residual weight) and continue peeling.
At termination, the remaining constraints produce a linear system in the guessed variables, which is solved by GE.

\subsection{Relation to ML decoding on surface codes}\label{sec:surface_proof}

\begin{definition}[Surface code]\label{def:sc}
Let $G=(V,E,F)$ be a square $L\times L$ lattice with boundary. Each edge $e\in E$ represents a qubit, each vertex $v\in V$ defines the $Z$-type stabilizer $Z_v$ and each face $f\in F$ defines the $X$-type stabilizer $X_f$ where
\[Z_v=\prod_{e\ni v}Z_e, \ X_f=\prod_{e\in\partial f}X_e .\]

The surface code is the stabilizer code with stabilizer group $S=\langle Z_v, X_f : v\in V, f\in F\rangle$.
\end{definition}

\begin{lemma}[Known-column weight-2 elimination identifies all fully-erased $X$-stabilizers on a surface code]
\label{lem:rule1_finds_all_fully_erased}
Consider the surface code of Definition~\ref{def:sc}.  Exhaustively applying known-column weight-2 elimination to $H_X$ guided by $\cK$ produces a reduced matrix
$\widetilde H_X$ with the following property: linear combinations of the rows that have zero support on $\cK$ form a basis for the subspace of fully erased $X$-stabilizers.
\end{lemma}

\ifarxiv \begin{proof} See Appendix~\ref{app:dualpeeling-omitted-proofs} for the proof. \end{proof} \fi
\begin{theorem}[Known-column weight-2 elimination followed by primal peeling is ML on surface codes over erasure]
\label{thm:ml_surface_code_rule1_short}
\end{theorem}
\begin{proof}
Let us focus on the linear system $H_{Z,\cE}x_{\cE}=s_x$.
For a surface code, $H_Z$ is a (boundary-modified) vertex--edge incidence matrix, hence the null space of $H_{Z,\cE}$ contains all cycles in the erased subgraph.
In other words, $H_{Z,\cE} x_{\cE}' = 0$ if and only if the erased edges associated with $\mathrm{supp}(x_{\cE}')$ form a disjoint union of cycles.  
These cycles decompose into cycles generated by (i) fully-erased $X$-stabilizers (i.e., face boundaries), and (ii) erased $X$-logical operators.

By Lemma~\ref{lem:rule1_finds_all_fully_erased}, exhaustive known-column weight-2 elimination
reveals all fully erased $X$-stabilizers supported on $\cE$.  The first type of cycles (i.e., fully erased $X$-stabilizers) are removed by fixing one erased qubit in the support of each independent stabilizer.
Thus, any residual nonzero portion of $x_{\cE}'$ must be an erased $X$-logical operator.
If there is no erased $X$-logical operator, then the linear system (after fixing bits in erased stabilizers) has a unique solution.
Thus, the residual erased subgraph is acyclic and hence peelable. So, primal peeling uniquely determines all entries of $x_{\cE}$ from $H_{Z,\cE}x_{\cE}=s_x$.
Conversely, if an erased $X$-logical operator exists, then ML decoding must fail.
Hence the algorithm succeeds on the same set of erasures as ML decoding.
\end{proof}

\section{Numerical Results}
\label{sec:results}

We now present numerical results for the stabilizer-assisted inactivation decoder over a range of code families on the erasure channel. We consider four decoders:
\begin{itemize}
  \item primal peeling decoder;
  \item dual peeling followed by primal peeling;
  \item peeling with hard guessing when stuck;
  \item inactivation decoding;
  \item inactivation decoding with stabilizer assistance.
\end{itemize}

\begin{figure}[t]
    \hspace*{-7mm}
    \scalebox{0.60}{
    \input{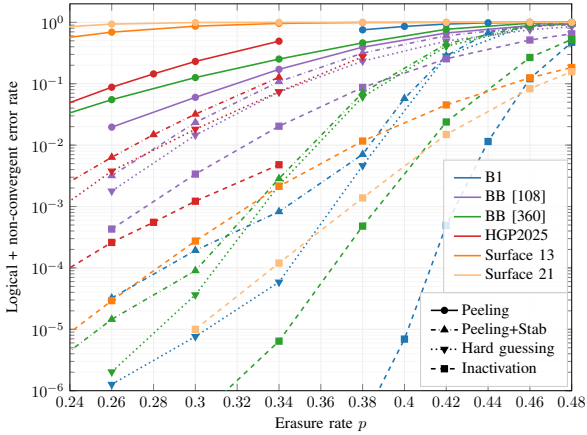}
    }
    \vspace{-3mm}
    \caption{Decoder comparison across code families on the quantum erasure channel: logical failure rate versus erasure rate $p$ for peeling, inactivation, and the two stabilizer-assisted variants.}
    \label{fig:errorrate}
    \vspace{-4mm}
\end{figure}

\subsection{Logical Error Rates and ML Optimality}

Figure~\ref{fig:errorrate} plots the logical failure probability as a function of the erasure probability $p$ for several code families, comparing the ML inactivation decoder, the primal peeling decoder, and the combined dual-primal peeling decoder.

For readability, we omit the curve for “dual peeling followed by inactivation,” since it is maximum-likelihood and therefore coincides with (simple) inactivation decoding across all code families. In addition, Section~\ref{sec:surface_proof} shows that dual-primal peeling is ML for surface codes, so its curve is likewise omitted.

\subsection{Effect of Dual Peeling}

In Figure~\ref{fig:errorrate}, dual peeling yields a substantial improvement over primal peeling across all code families. We also observe that the combined dual-primal peeling decoder, while not matching ML performance, requires only $O(n)$ operations, in contrast to inactivation decoding, which ultimately uses Gaussian elimination on the matrix $C$ with complexity $O(|\cI|^3)$.

Dual peeling also reduces the size of the matrix that must be inverted. This reduction is evident in the low-erasure regime, but remains significant even at higher erasure rates
though the effect is weaker for other families of codes.

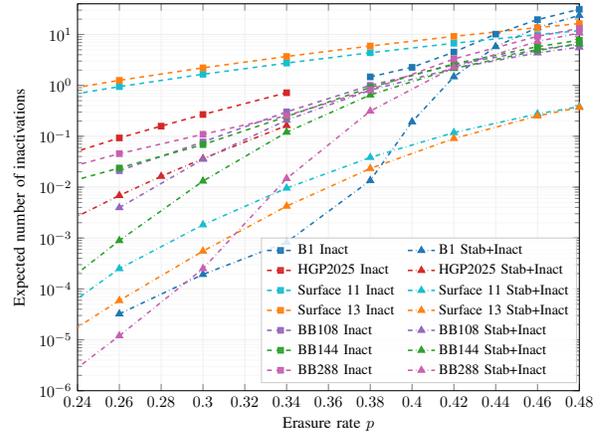
\begin{figure}[t]
    \hspace{0mm}
    \scalebox{0.60}{
\begin{tikzpicture}
\definecolor{tabblue}{RGB}{31,119,180}
\definecolor{tabred}{RGB}{214,39,40}
\definecolor{taborange}{RGB}{255,127,14}
\definecolor{tabcyan}{RGB}{23,190,207}
\definecolor{tabgreen}{RGB}{44,160,44}
\definecolor{tabpurple}{RGB}{148,103,189}
\definecolor{tabbrown}{RGB}{200,96,175}
\definecolor{lightgray204}{RGB}{204,204,204}
\definecolor{lightgray}{RGB}{240,240,240}

\pgfplotsset{
  codeBOne/.style={draw=tabblue},
  codeHGP/.style={draw=tabred},
  codeSurfEleven/.style={draw=tabcyan},
  codeSurfThirteen/.style={draw=taborange},
  codeBBOneOhEight/.style={draw=tabpurple},
  codeBBOneFourFour/.style={draw=tabgreen},
  codeBBTwoEightEight/.style={draw=tabbrown},
  codeDefault/.style={draw=black},
  decInact/.style={dashed, line width=1pt, mark=square*, mark size=1.6pt},
  decStabInact/.style={dashdotted, line width=1pt, mark=triangle*, mark size=2.2pt}
}

\begin{axis}[
  label style={font=\normalsize},
  xlabel={Erasure rate $p$},
  ylabel={Expected number of inactivations},
  width=5in,
  height=4in,
  ymode=log,
  unbounded coords=jump,
  grid=both,
  major grid style={draw=lightgray, opacity=1},
  minor grid style={draw=lightgray, opacity=0.5},
  ticklabel style={font=\normalsize},
  label style={font=\normalsize},
  title style={font=\large},
  xmin=0.24, xmax=0.48,
  ymin=1e-06, ymax=4e1,
  mark layer=above,
  legend columns=2,
  legend cell align={left},
  legend style={
    at={(0.995,0.02)},
    anchor=south east,
    draw=lightgray204,
    fill opacity=0.85,
    text opacity=1,
    font=\small,
    row sep=1pt,
    /tikz/every even column/.append style={column sep=6pt}
  },
  legend image post style={line width=1pt},
  scale=1.0
]

\addplot+[
  codeBOne,
  decInact,
  mark options={solid, draw=tabblue, fill=tabblue}
] table [x=x, y=y, col sep=comma] {
x,y
0.38,1.4685975
0.4,2.260899807692307
0.42,4.50768
0.44,10.151836
0.46,19.564184
0.48,31.316024
};
\addlegendentry{B1 Inact}

\addplot+[
  codeBOne,
  decStabInact,
  mark options={solid, draw=tabblue, fill=tabblue}
] table [x=x, y=y, col sep=comma] {
x,y
0.26,3.233333333333333e-05
0.3,0.0001926666666666667
0.34,0.0008266666666666666
0.38,0.0134884375
0.4,0.1922228846153846
0.42,1.474488
0.44,5.77898
0.46,13.703936
0.48,23.717984
};
\addlegendentry{B1 Stab+Inact}

\addplot+[
  codeHGP,
  decInact,
  mark options={solid, draw=tabred, fill=tabred}
] table [x=x, y=y, col sep=comma] {
x,y
0.14,0.001243777777777778
0.18,0.0069475
0.22,0.028093
0.26,0.0927625
0.28,0.1581758064516129
0.3,0.2677225
0.34,0.7169
};
\addlegendentry{HGP2025 Inact}

\addplot+[
  codeHGP,
  decStabInact,
  mark options={solid, draw=tabred, fill=tabred}
] table [x=x, y=y, col sep=comma] {
x,y
0.14,1.166666666666667e-05
0.18,0.0001225
0.22,0.001053
0.26,0.006845
0.28,0.01629354838709677
0.3,0.0363775
0.34,0.1623775
};
\addlegendentry{HGP2025 Stab+Inact}

\addplot+[
  codeSurfEleven,
  decInact,
  mark options={solid, draw=tabcyan, fill=tabcyan}
] table [x=x, y=y, col sep=comma] {
x,y
0.18,0.2412466666666667
0.22,0.503355
0.26,0.94452
0.3,1.65291
0.34,2.740453333333333
0.38,4.376233333333333
0.42,6.722333333333333
0.46,9.921329999999999
0.48,11.87007
};
\addlegendentry{Surface 11 Inact}

\addplot+[
  codeSurfEleven,
  decStabInact,
  mark options={solid, draw=tabcyan, fill=tabcyan}
] table [x=x, y=y, col sep=comma] {
x,y
0.18,1.666666666666667e-06
0.22,1.666666666666667e-05
0.26,0.00025
0.3,0.00182
0.34,0.009573333333333333
0.38,0.03838
0.42,0.1181066666666667
0.46,0.27526
0.48,0.38198
};
\addlegendentry{Surface 11 Stab+Inact}

\addplot+[
  codeSurfThirteen,
  decInact,
  mark options={solid, draw=taborange, fill=taborange}
] table [x=x, y=y, col sep=comma] {
x,y
0.18,0.3161582222222222
0.22,0.6638848888888889
0.26,1.256050222222222
0.3,2.213034
0.34,3.698305333333333
0.38,5.934980666666666
0.42,9.176641333333333
0.46,13.64338866666667
0.48,16.37886933333333
};
\addlegendentry{Surface 13 Inact}

\addplot+[
  codeSurfThirteen,
  decStabInact,
  mark options={solid, draw=taborange, fill=taborange}
] table [x=x, y=y, col sep=comma] {
x,y
0.18,0
0.22,5.333333333333334e-06
0.26,5.911111111111111e-05
0.3,0.0005533333333333333
0.34,0.004229333333333333
0.38,0.02324933333333333
0.42,0.090374
0.46,0.2505166666666667
0.48,0.3681333333333333
};
\addlegendentry{Surface 13 Stab+Inact}

\addplot+[
  codeBBOneOhEight,
  decInact,
  mark options={solid, draw=tabpurple, fill=tabpurple}
] table [x=x, y=y, col sep=comma] {
x,y
0.26,0.02104933333333333
0.3,0.07742733333333333
0.34,0.3016293333333334
0.38,0.998884
0.42,2.537714
0.46,5.034072
0.48,6.582673333333333
};
\addlegendentry{BB108 Inact}

\addplot+[
  codeBBOneOhEight,
  decStabInact,
  mark options={solid, draw=tabpurple, fill=tabpurple}
] table [x=x, y=y, col sep=comma] {
x,y
0.26,0.003959333333333333
0.3,0.03534066666666667
0.34,0.2072486666666667
0.38,0.8027153333333333
0.42,2.161122666666667
0.46,4.344790666666666
0.48,5.643398
};
\addlegendentry{BB108 Stab+Inact}

\addplot+[
  codeBBOneFourFour,
  decInact,
  mark options={solid, draw=tabgreen, fill=tabgreen}
] table [x=x, y=y, col sep=comma] {
x,y
0.22,0.008354666666666666
0.26,0.02380333333333333
0.3,0.06858400000000001
0.34,0.2437
0.38,0.9082773333333334
0.42,2.650231333333333
0.46,5.753414666666667
0.48,7.732769333333334
};
\addlegendentry{BB144 Inact}

\addplot+[
  codeBBOneFourFour,
  decStabInact,
  mark options={solid, draw=tabgreen, fill=tabgreen}
] table [x=x, y=y, col sep=comma] {
x,y
0.22,4.666666666666666e-05
0.26,0.000892
0.3,0.01315066666666667
0.34,0.1201286666666667
0.38,0.64861
0.42,2.148179333333333
0.46,4.846968666666666
0.48,6.518218666666667
};
\addlegendentry{BB144 Stab+Inact}

\addplot+[
  codeBBTwoEightEight,
  decInact,
  mark options={solid, draw=tabbrown, fill=tabbrown}
] table [x=x, y=y, col sep=comma] {
x,y
0.22,0.01641733333333333
0.26,0.045464
0.3,0.1095086666666667
0.34,0.2535013333333334
0.38,0.8179346666666667
0.42,3.305116666666667
0.46,9.189726666666667
0.48,13.15619933333333
};
\addlegendentry{BB288 Inact}

\addplot+[
  codeBBTwoEightEight,
  decStabInact,
  mark options={solid, draw=tabbrown, fill=tabbrown}
] table [x=x, y=y, col sep=comma] {
x,y
0.22,6.666666666666667e-07
0.26,1.2e-05
0.3,0.00025
0.34,0.01476866666666667
0.38,0.3120786666666667
0.42,2.269012
0.46,7.307929333333333
0.48,10.68673333333333
};
\addlegendentry{BB288 Stab+Inact}

\end{axis}
\end{tikzpicture}
    }
    \vspace{-3mm}
    \caption{Comparison of the expected size of the inverted matrix for inactivation decoding, with and without stabilizer peeling, versus erasure rate $p$.}
    \label{fig:histogram}
    \vspace{-0.5mm}
\end{figure}

\vspace{-0mm}
\section{Conclusion and Future Work}
\label{sec:conclusion}

We presented a stabilizer-assisted inactivation decoder for CSS/QLDPC codes on the quantum erasure channel that combines (i) a dual peeling preprocessing step on the stabilizer matrix to find fully erased stabilizers and (ii) standard primal peeling + inactivation on the induced linear system. The key insight is that stabilizers supported entirely on the erasure set correspond to degeneracies that can be fixed “for free,” so identifying them before inactivation reduces the size of the inverted matrix, therefore reducing the decoding complexity.

Across the code families considered, dual peeling improves on low-complexity peeling-based decoding and inactivation decoding by reducing the dimension of the residual system that must be solved. In addition, for surface codes the combined dual-primal peeling procedure results in ML performance on the erasure channel.
Several directions for future work remain:
\begin{itemize}
  \item extend stabilizer-assisted decoding beyond the BEC to bit-flip channels, possibly through BP-based algorithms;
  \item compare with more complicated non-erasure decoders such as BPGD on the erasure channel.
\end{itemize}

\vspace{-0mm}
\section*{Acknowledgements}

We also acknowledge the use of generative AI tools for tasks related to programming, editing, and proofreading.

\newpage

\bibliographystyle{IEEEtran}
\bibliography{refs}

\ifarxiv
\appendices

\section{Notation}
\label{app:notation}

We note here some notational conventions we use but have not explicitly defined.

\begin{itemize}
  \item \emph{Subvectors and submatrices.}
  For an index set $\cS \subseteq \{1,\dots,n\}$ and a vector $v\in\mathbb{F}_2^n$, the notation
  $v_{\cS}$ denotes the restriction of $v$ to coordinates in $\cS$.
  For a matrix $M$ with $n$ columns, $M_{:,\cS}$ (abbreviated $M_{\cS}$)
  denotes the submatrix formed by columns in $\cS$; similarly $M_{\cR,:}$ restricts rows.
  Complements are denoted $\cS^c$.

  \item \emph{Weight.}
  For a vector $v$, we denote by $\wt(v)$ its Hamming weight (i.e., the number non-zero entries).

  \item \emph{Inactivation block form.}
  When erased columns are permuted into \emph{active} and \emph{inactivated} groups,
  the erased submatrix may be written in block form
  \[
    H_{\cE} \ \equiv\ 
    \begin{bmatrix}
      A & D\\
      B & C
    \end{bmatrix},
  \]
  the right block-column corresponds to inactivated (guessed) variables.
  After the peeling and guessing stage of inactivation decoding finishes but before Gaussian elimination, block $A$ will be the zero matrix and block $B$ will be a permutation of the identity matrix.
\end{itemize}

\section{Omitted Proofs}
\label{app:dualpeeling-omitted-proofs}



\subsection{ Lemma~\ref{lem:rule1_finds_all_fully_erased}: Known-column weight-2 elimination identifies all fully-erased $X$-stabilizers on a surface code}

\label{app:proof_lem_rule1_finds_all_fully_erased}

\begin{proof}
Let $\Phi$ be a partition of $F$ into disjoint contiguous subsets.  For each part $C\in\Phi$ define the $X$-stabilizer $X_C:=\prod_{f\in C}X_f$ and let $h_C\in\mathbb{F}_2^{|E|}$ be the $\mathbb{F}_2$-sum of
$\{\partial f:f\in C\}$.

\emph{Inductive hypothesis:} In the partition $\Phi_t$, every known interior edge $j\in\cK$ is either
(i) a boundary edge between two distinct parts of $\Phi_t$, or (ii) internal to a part of $\Phi_t$ and hence canceled
from the corresponding $h_C$.  

Initialize $\Phi_0=\{\{f\}:f\in F\}$. Hypothesis holds for $\Phi_0$ because every edge is at the boundary of two subsets of faces in our partition and none of the subsets have any interior edges. We proceed by induction over the number of elimination merges. Assume the hypothesis for $\Phi_t$. If $j\in\cK$ is a boundary edge between parts $C_1,C_2\in\Phi_t$, applying known-column weight-2 elimination replaces the two rows by their $\mathbb{F}_2$-sum.
In other word, it replaces $X_{C_1},X_{C_2}$ by
$X_{C_1}X_{C_2}=X_{(C_1\cup C_2) \setminus \{j\}}$ because $j$ becomes internal to $C_1\cup C_2$ and cancels out, while all other known
interior edges keep property (i) or (ii).  Hence the hypothesis holds for $\Phi_{t+1}$.

After elimination, no known interior edge separates two parts.
Therefore, each row $h_C$ has zero support on $\cK$ exactly when $X_C$ is fully erased.  Conversely, any fully erased $X$-stabilizer
$u^\T H_X$ is the $\mathbb{F}_2$-boundary of a union of faces with no known-edges on its boundary and $u^\T H_X$ lies in the span of the final set of $h_C$'s.
\end{proof}

\fi 

\section{Complexity Analysis}

\subsection{Dual peeling complexity with Rule-2}

Here, we consider the complexity of dual peeling using only Rule-2.
Recall the the columns of $H_X$ are partitioned into the known set $\cK$ and erased set $\cE$, and write each row as
$H_{X,i} = [\,H_{X,i,\cK}\;\; H_{X,i,\cE}\,]$ after reordering columns.
In the dual peeling algorithm restricted to weight-$1$ moves (Rule-2), the control logic depends only on the
\emph{known} submatrix $H_{X,\cK}$ through the known-degree
\[
d_{\cK}(i) \triangleq \wt\!\big(H_{X,i,cK}\big).
\]
At each step, one selects a row $i$ with $d_{\cK}(i)=1$ and unique known neighbor $j\in \cK$, and performs row
eliminations $H_{X,k}\leftarrow H_{X,k}\oplus H_{X,i}$ for all other rows $k$ incident to $j$.
Implementing this phase using standard adjacency lists for the sparse bipartite incidence graph induced by $H_{X,\cK}$
yields time essentially linear in the number of non-zero elements in $H_{X,\cK}$:
each known-side edge is removed at most once, and degree updates can be charged to a constant number of
adjacency operations.  Upon termination, the number of \emph{fully erased} stabilizer generators is simply
\[
s \triangleq \big|\{\, i : d_{\cK}(i)=0 \,\}\big|,
\]
which can therefore be computed without inspecting the erased portion $H_{X,\cE}$.

\paragraph{Reducing complexity by deferring work}
A naive implementation applies each elimination $H_{X,k}\leftarrow H_{X,k}\oplus H_{X,i}$ to the full row
$[H_{X,k,\cK}\;H_{X,k,\cE}]$, which can induce substantial fill-in in $H_{X,\cE}$ even though the known-side graph
remains sparse.
To control this cost, we treat the known-side peeling as applying elementary row operations to $H_{X}$ but performing them only on $H_{X,\cK}$ and track for future application to $H_{X,\cE}$.
Let $T$ denote the transform implied by the tracked operations.
Correctness is immediate: applying the same transform $T$ to the full matrix
$[H_{X,\cK}\;H_{X,\cE}]$ either eagerly or lazily produces identical final rows (i.e., the same
set of fully erased generators).  From a complexity perspective, this decouples the overall cost into (i) the peeling phase,
which is linear in the number of edges ones in $H_{X,\cK}$, and (ii) a \emph{replay} phase in which one evaluates the action of $T$ as needed.  In particular, the replay cost depends on how many columns $H_{X,\cE}$ we choose to compute and this could be much smaller than $|\cE|$.

\paragraph{Random-column replay}
Rather than replaying $T$ on all erased columns in $\cE$, one may restrict attention to a random subset
$R \subseteq \cE$ whose size scales with the number of fully erased generators.
Concretely, after peeling on $H_{X,\cK}$ yields the fully erased row set $S=\{i:d_{\cK}(i)=0\}$ with $|S|=s$, choose a random subset $R\subseteq \cE$ (say with with $|R|\approx 2s$), and replay the logged row operations only on the implied column submatrix of $H_{X,\cE}$.
The goal is not to reconstruct all erased supports, but only to identify a full set of bits that can be fixed to 0 via stabilizer pruning.
The row operations necessary for pruning would also be applied only to this submatrix.

The basic idea actually allows Rule-1 dual peeling as well but, in that case, we cannot argue that the first stage complexity is linear in the number of non-zero elements in $H_{X,\cK}$.
Instead, we believe additional log factors may be required and we leave this analysis as future work.

\subsection{Faster Implementation for Surface Code}

To minimize complexity, one can implement the Rule-1 dual peeling step on a surface code by working directly with the dual lattice rather than performing explicit sparse row additions on $H_X$.
Recall that each qubit corresponds to an edge $e\in E$ and each $X$-stabilizer corresponds to a face $f\in F$.
Moreover, interior edges are incident to exactly two faces.
For every \emph{known} qubit $e\in \cK$ that is incident to two faces $f_1,f_2\in F$, the weight-2 elimination move that would cancel the column of $e$ in $H_{X,\cK}$ can be realized as contracting the dual edge $(f_1,f_2)$.  We implement these contractions using the disjoint-set union data structure over the face set $F$, where each contractible known edge triggers a union operation.

This implementation avoids intermediate fill-in that can occur if one instead uses row-xors on sparse row supports.  With union-by-rank (or union-by-size) and path compression, a sequence of $m$ union/find operations on $|F|$ elements runs in $O(m\,\alpha(|F|))$ time, where $\alpha(\cdot)$ is the inverse Ackermann function~\cite{Tarjan1975UnionFind,CLRS2009}.  Since a surface code has $|F|=\Theta(n)$ faces and at most $\Theta(n)$ interior known edges eligible for contraction, the full elimination pass runs in $O(n\,\alpha(n))$ time and $O(n)$ space.  In practice, $\alpha(n)$ grows so slowly that it it is treated as constant less than 5 for all feasible problem sizes.
Thus, the preprocessing is effectively linear-time with small overhead.

After contracting all eligible known dual edges, each union--find component represents a reduced $X$-check obtained by xoring the faces in that component, exactly mirroring exhaustive known-column weight-2 elimination.  To identify fully-erased $X$-stabilizers, we perform a single additional scan over edges to mark which components retain any support on $K$ (e.g., via boundary edges incident to a single face or other non-contractible known support); components with zero support on $K$ correspond to fully-erased $X$-stabilizers.  We also record, for each such component, a representative erased qubit $e\in E$ on which we can impose an arbitrary gauge choice (e.g., fix $x_e=0$) to remove the associated cycle degeneracy prior to primal peeling.

Finally, we run standard queue-based primal peeling on the syndrome system $H_{Z,E}x_E=s_x$.  For surface codes, $H_Z$ is a (boundary-modified) vertex--edge incidence matrix, so each peeling step eliminates an erased edge when a vertex constraint has residual erased-degree one, updating at most a constant number of neighboring degrees because vertex degrees are bounded.  Hence primal peeling runs in $O(|E|)$ time and $O(n)$ space.  Overall, the ML erasure decoder for surface codes given by Rule-1 dual peeling plus primal peeling has time complexity $O(n\,\alpha(n)+|E|)=O(n\,\alpha(n))$ and memory $O(n)$.

\end{document}